\documentclass{LMCS}

\def\dOi{11(1:16)2015}
\lmcsheading%
{\dOi}
{1--27}
{}
{}
{Nov.~12, 2014}
{Mar.~31, 2015}
{}

\ACMCCS{[{\bf Theory of computation}]: Semantics and
  reasoning---Program reasoning---Program verification; [{\bf Software
      and its engineering}]: Software organization and
  properties---Software functional properties---Formal
  methods---Software verification}

\subjclass{
	D.2.4 [Software Engineering]: Software/Program Verification---Formal methods;
	F.3.1 [Logics and Meanings of Programs]: Specifying and Verifying and Reasoning about Programs---Mechanical verification%
}

\usepackage[english]{babel}
\usepackage[utf8]{inputenc}
\usepackage{hyperref}
\usepackage{wrapfig}

\usepackage{amsmath}
\usepackage{amssymb}

\usepackage{tikz}
\usetikzlibrary{automata,arrows}

\usepackage{listings}
\lstdefinelanguage{pseudo}{
	morekeywords={if, then, else, while, foreach, do, let, return},
	sensitive=false,
	morecomment=[l]{//},
	morecomment=[s]{/*}{*/},
	morestring=[b]",
}
\AtBeginDocument{

}

\usepackage{aliascnt}

\newcommand{\mynewtheorem}[2]{
	\newaliascnt{#1}{dummy}
	\newtheorem{#1}[#1]{#2}
	\aliascntresetthe{#1}
	\expandafter\def\csname #1autorefname\endcsname{#2}
}
\newcounter{dummyBackup}

\theoremstyle{plain}
\mynewtheorem{theorem}{Theorem}
\mynewtheorem{lemma}{Lemma}
\mynewtheorem{corollary}{Corollary}
\theoremstyle{definition}
\mynewtheorem{definition}{Definition}
\mynewtheorem{remark}{Remark}
\mynewtheorem{example}{Example}

\usepackage{etoolbox}
\AtBeginEnvironment{example}{
	\renewcommand{\qed}{$\Diamond$}
}
\AtEndEnvironment{example}{
	\hfill\qed
}

\renewcommand{\qedhere}{\eqno\qEd}
\newcommand{\noqed}{\renewcommand{\qed}{}}

\newcommand{\stemt}{{\scriptstyle\mathrm{STEM}}} 
\newcommand{\loopt}{{\scriptstyle\mathrm{LOOP}}} 
\newcommand{\prog}{\mathbf{P}}                   
\newcommand{\T}{{\scriptstyle\mathrm{T}}}        
\newcommand{\abovebelow}[2]{(^{#1}_{#2})}        
\newcommand{\tr}[1]{#1^T}                        

\begin{document}

\title[Ranking Templates for Linear Loops]{\quad\\\quad\\Ranking Templates for Linear Loops\rsuper*}

\author[J.~Leike]{Jan Leike\rsuper a}
\address{{\lsuper a}The Australian National University}
\email{jan.leike@anu.edu.au}

\author[M.~Heizmann]{Matthias Heizmann\rsuper b}
\address{{\lsuper b}University of Freiburg}
\email{heizmann@informatik.uni-freiburg.de}
\thanks{{\lsuper b}This work is supported by the
 German Research Council (DFG) as part of the Transregional Collaborative
 Research Center ``Automatic Verification and Analysis of Complex Systems''
 (SFB/TR14 AVACS)}

\keywords{
	Linear lasso program,
	linear loop program,
	termination,
	linear ranking template,
	well-founded relation,
	multiphase ranking function,
	nested ranking function,
	piecewise ranking function,
	lexicographic ranking function,
	parallel ranking function,
	Farkas' lemma,
	Motzkin's theorem}


\titlecomment{{\lsuper*}An earlier version of this paper appeared in TACAS 2014~\cite{LH14}.}


\begin{abstract}
We present a new method for the constraint-based synthesis
of termination arguments for linear loop programs
based on \emph{linear ranking templates}.
Linear ranking templates are parameterized, well-founded relations such that
an assignment to the parameters gives rise to a ranking function.
Our approach generalizes existing methods and
enables us to use templates for
many different ranking functions with affine-linear components.
We discuss templates for
\emph{multiphase}, \emph{nested}, \emph{piecewise}, \emph{parallel}, and \emph{lexicographic} ranking functions.
These ranking templates can be combined to form more powerful templates.
Because these ranking templates require both strict and non-strict inequalities,
we use Motzkin's transposition theorem instead of Farkas' lemma
to transform the generated $\exists\forall$-constraint
into an $\exists$-constraint.
\end{abstract}

\maketitle


\section{Introduction}
\label{sec:introduction}

The scope of this work is the constraint-based synthesis of termination arguments.
In our setting,
we consider \emph{linear loop programs},
which are specified by
a boolean combination of affine-linear inequalities over the program variables.
This allows for both,
deterministic and non-deterministic updates of the program variables.
An example of a linear loop program is given in
\autoref{fig:multiphase-introduction}.

\begin{figure}[t]
\begin{center}
\begin{minipage}{35mm}
\begin{lstlisting}
while ($q > 0$):
    $q$ := $q - y$;
    $y$ := $y + 1$;
\end{lstlisting}
\end{minipage}
\begin{minipage}{35mm}
\begin{align*}
q &> 0  \\
\land\; q' &= q - y \\
\land\; y' &= y + 1
\end{align*}
\end{minipage}
\end{center}
\caption{
A linear loop program given as program code (left) and
as a formula defining a binary relation (right).
}\label{fig:multiphase-introduction}
\end{figure}

Usually, linear lasso programs do not occur as stand-alone programs.
Instead, they are used as a finite representation of
an infinite path in a control flow graph.
For example, in (potentially spurious) counterexamples in termination analysis~\cite{CPR06,BCF13,HLNR10,KSTTW08,KSTW10,PR04TI,PodelskiR05,HHP14},
non-termination analysis~\cite{GHMRX08},
stability analysis~\cite{CFKP11,PW07},
or cost analysis~\cite{AAGP11,GZ10}.

We introduce the notion of \emph{linear ranking templates}
(\autoref{sec:templates}).
These are parameterized relations specified by linear inequalities such that
any assignment to the parameters yields a well-founded relation.
This notion is general enough to encompass all existing methods for linear loop programs
that use constraint-based synthesis of ranking functions of various kinds
(see \autoref{sec:related-work} for an assessment).
Moreover, ours is the first method for synthesis of lexicographic ranking functions
that does not require a mapping between loop disjuncts and lexicographic components.

In this paper we present the following linear ranking templates.
\begin{itemize}
\item The \emph{multiphase ranking template} specifies a ranking function
that proceeds through a fixed number of phases in the program execution.
Each phase is ranked by an affine-linear function; when this function
becomes non-positive, we move on to the next phase
	(\autoref{ssec:rt-multiphase}).
	We call such a ranking function a multiphase ranking function.
\item The \emph{nested ranking template} specifies a ranking function
	that is a special case of a multiphase ranking function
	(\autoref{ssec:rt-nested}).
	In contrast to the multiphase ranking template,
	the nested ranking template requires only linear constraint solving.
\item The \emph{piecewise ranking template} specifies a ranking function
that is a piecewise affine-linear function
with affine-linear predicates to discriminate between the pieces (\autoref{ssec:rt-pw}).
\item The \emph{lexicographic ranking template} specifies a lexicographic ranking function
that corresponds to a tuple of affine-linear functions together with a lexicographic ordering on the tuple (\autoref{ssec:rt-lex}).
\item The \emph{parallel ranking template} targets programs that
	have to complete a finite number of independent tasks
	with no predetermined order (\autoref{ssec:rt-parallel}).
\end{itemize}
Furthermore, our linear ranking templates can be used as a `construction kit'
for composing linear ranking templates
that enable more complex ranking functions (\autoref{sec:composition-of-templates}).
Thus, variations on the linear ranking templates presented here can be used
and completely different templates could be conceived.

Our method is described in \autoref{sec:synthesizing-ranking-functions}
and can be summarized as follows.
The input is a linear loop program as well as a linear ranking template.
From these we construct a constraint on the parameters of the template.
This constraint is a quantified nonlinear SMT formula.
With \hyperref[thm:Motzkin]{Motzkin's theorem}~\cite{Schrijver99} we transform the constraint
into a purely existentially quantified constraint.
This $\exists$-constraint is then passed to an SMT solver
which checks its satisfiability.
A positive result implies that the program terminates.
Furthermore, a satisfying assignment yields a ranking function,
which constitutes a termination argument for the given linear loop program.

Related approaches invoke Farkas' lemma for the transformation into $\exists$-constraints
\cite{ADFG10,BMS05linrank,BMS05polyrank,CSS03,HHLP13,PR04,Rybalchenko10,SSM04}.
Several of our ranking templates contain
both strict and non-strict inequalities, yet only non-strict inequalities can be transformed using Farkas' lemma.
We solve this problem by introducing the use of
\hyperref[thm:Motzkin]{Motzkin's Transposition Theorem}, a generalization of Farkas' lemma.
As a side effect, this also enables both strict and non-strict inequalities in the program syntax.
To our knowledge, all of the aforementioned methods can be 
extended to programs with strict inequalities
if \hyperref[thm:Motzkin]{Motzkin's theorem} is applied instead of Farkas' lemma.
%

Our method is complete in the following sense.
If there is a ranking function of the form specified by the given linear ranking template,
then our method will discover this ranking function.
In other words, the existence of a solution is never lost in the process of transforming the constraint.

In contrast to some related methods~\cite{HHLP13,PR04}
the constraint we generate is not linear,
but rather a nonlinear algebraic constraint.
Theoretically, this constraint can be decided in exponential time~\cite{GV88}.
Much progress on nonlinear SMT solvers has been made and
present-day implementations routinely solve nonlinear constraints of various sizes~\cite{JM12}.

A related setting to linear loop programs are \emph{linear lasso programs}
(see \autoref{fig:lasso}).
These consist of a linear loop program and a program stem,
both of which are specified by
boolean combinations of affine-linear inequalities over the program variables.
Our method can be extended to linear lasso programs
through the addition of affine-linear inductive invariants,
analogously to related approaches~\cite{BMS05linrank,CSS03,HHLP13,SSM04}
(\autoref{sec:linear-lasso-programs}).

In this work, we consider variables with values in the rational or real numbers.
Our method can be applied directly to integer programs,
but in this case our completeness result does not hold.
However, if we compute the integral hull of transition relations
analogously to \cite{CKRW13,HHLP13},
we obtain the same completeness result for integer-valued linear loop programs
as for rational- and real-valued linear loop programs.

This journal article is an extension of a conference paper~\cite{LH14}.
In the conference paper,
we introduced the notion of ranking templates and showed how to
solve them using Motzkin's theorem
(\autoref{sec:ranking-templates} and \autoref{sec:synthesizing-ranking-functions}).
We discussed the multiphase, the piecewise,
and the lexicographic ranking template.
The main additions in this article are
the nested and parallel ranking template,
the composition of ranking templates, and
the extension of our method to linear lasso programs,
as well as additional examples.


\section{Preliminaries}
\label{sec:preliminaries}

In this paper we use $\mathbb{K}$ to denote a field that is either the rational numbers $\mathbb{Q}$ or
the real numbers $\mathbb{R}$.
%

\subsection{Set Theory}
\label{ssec:set-theory}

We use the the following notions from set theory~\cite{Jech06}.
A set $X$ is \emph{transitive} iff every element of $X$ is a subset of $X$.
A relation $R \subseteq X \times X$ is \emph{well-founded} iff
every non-empty subset of $X$ has an $R$-minimal element.

\begin{definition}[{Ordinal Number~\cite[Def.\ 2.10]{Jech06}}]
\label{def:ordinal-number}
A set $\alpha$ is an \emph{ordinal number} (an \emph{ordinal}) iff
it is transitive and
$\in$ (`element-of') is a well-founded total order on $\alpha$.
\end{definition}

The ordinal numbers are a method of counting `beyond infinity'.
The smallest ordinal is the empty set,
and for every ordinal $\alpha$
there is a unique successor ordinal $\alpha \cup \{ \alpha \}$.
The finite ordinals coincide with the natural numbers,
therefore we use them interchangeably.
The smallest infinite ordinal is denoted by $\omega$.
Ordinals can be added, multiplied and exponentiated,
but in general these operations are not commutative.

We use the following theorem which allows us to define functions recursively.

\begin{theorem}[{Recursion Theorem~\cite[Thm.\ 6.11]{Jech06}}]
\label{thm:recursion}
Let $R$ be a well-founded relation on the set $X$ and
let $G$ be a function on sets.
Then there is a unique function $F$ with domain $X$ such that
for every $x \in X$,
\[
F(x) = G(x, F|_{\{ z \in X \mid (z, x) \in R \}}),
\]
where $F|_Y$ is the restriction of the function $F$ to the domain $Y$.
\end{theorem}

\subsection{Linear Loop Programs}
\label{ssec:loops}

In this work, we consider programs that consist of a single loop.
We use binary relations over the program's states to define its transition relation.

We denote by $x$
the vector of $n$ variables $(x_1, \ldots, x_n)^T \in \mathbb{K}^n$
corresponding to program states, and
by $x' = (x_1', \ldots, x_n')^T \in \mathbb{K}^n$
the variables of the next state.

\begin{definition}[Linear Loop Program]
\label{def:linear-loop}
A \emph{linear loop program} $\loopt(x, x')$ is a binary relation defined by a formula with the free variables $x$ and $x'$ of the form
\[
\bigvee_{i \in I} \big( A_i\abovebelow{x}{x'} \leq b_i \;\land\; C_i\abovebelow{x}{x'} < d_i \big)
\]
for some finite index set $I$, some matrices $A_i \in \mathbb{K}^{2n \times m_i}$,
$C_i \in \mathbb{K}^{2n \times k_i}$,
and some vectors $b_i \in \mathbb{K}^{m_i}$ and $d_i \in \mathbb{K}^{k_i}$.
The linear loop program $\loopt(x, x')$ is called \emph{conjunctive} iff
there is only one disjunct, i.e., $\#I = 1$.
\end{definition}
Geometrically the relation $\loopt$ corresponds to a union of convex polyhedra.

\begin{definition}[Termination]
\label{def:termination}
A linear loop program $\loopt(x, x')$ \emph{terminates} iff
the relation $\loopt(x, x')$ is well-founded.
\end{definition}

In general, the termination of linear loop programs is undecidable
because linear loop programs can be used to simulate counter machines.
An undecidability proof for the termination of linear lasso programs
is given in \cite[Thm.\ 3.18]{Leike13}.

\begin{example}
\label{ex:running1}
Consider the following program code.
\begin{center}
\begin{minipage}{48mm}
\begin{lstlisting}
while ($q > 0$):
    if ($y > 0$):
        $q$ := $q - y - 1$;
    else:
        $q$ := $q + y - 1$;
\end{lstlisting}
\end{minipage}
\end{center}
We represent this code using the following linear loop program:
\begin{align*}
&(q > 0 \;\land\; y > 0
  \;\land\; y' = y \;\land\; q' = q - y - 1) \\
\lor\; &(q > 0 \;\land\; y \leq 0
  \;\land\; y' = y \;\land\; q' = q + y - 1)
\end{align*}
This linear loop program is not conjunctive.
Furthermore, there is no infinite sequence of states $x_0, x_1, \ldots$ such that for all $i \geq 0$, the two successive states $(x_i, x_{i+1})$ are contained in the relation $\loopt$.
Hence the relation $\loopt(x, x')$ is well-founded and the linear loop program terminates.
We note that this linear loop program does not have a linear ranking function.
However, termination of this program can be proven using ranking functions
that we present in \autoref{ssec:rt-multiphase} and in \autoref{ssec:rt-nested}.
\end{example}


\section{Ranking Templates}
\label{sec:templates}

A \emph{ranking template} is a template for a well-founded relation.
More specifically, it is a parameterized formula defining a relation that is well-founded for all assignments to the parameters.
If we show that a given program's transition relation $\loopt$ is a subset of an instance of this well-founded relation, it must be well-founded itself and thus we have a proof for the program's termination.
Moreover, an assignment to the parameters of the template gives rise to a ranking function.
In this work, we consider ranking templates that can be encoded with linear arithmetic.

We call a formula whose free variables contain $x$ and $x'$ a \emph{relation template}.
Each free variable other than $x$ and $x'$ in a relation template is called \emph{parameter}.
Given an assignment $\nu$ to all parameter variables of a relation template $\T(x, x')$,
the evaluation $\nu(\T)$ is called an \emph{instantiation of the relation template $\T$}.
We note that each instantiation of a relation template $\T(x, x')$ defines a binary relation.

When specifying templates, we use parameter variables to define affine-linear functions.
For notational convenience, we write $f(x)$ instead of the term $\tr{s_f} x + t_f$,
where $s_f \in \mathbb{K}^n$ and $t_f \in \mathbb{K}$ are parameters.
We call $f$ an \emph{affine-linear function symbol}.

\begin{definition}[Linear Ranking Template]
\label{def:linear-rt}
Let $\T(x, x')$ be a relation template
with parameters $D$ and affine-linear function symbols $F$
that can be written as a boolean combination of atoms of the form
\begin{align*}
\sum_{f \in F} \big( \alpha_f \cdot f(x) + \beta_f \cdot f(x') \big)
+ \sum_{\delta \in D} \gamma_\delta \cdot \delta \;\rhd\; 0,
\end{align*}
where $\alpha_f, \beta_f, \gamma_\delta \in \mathbb{K}$ are constants and
$\rhd \in \{ \geq, > \}$.
We call $\T$ a \emph{linear ranking template} over $D$ and $F$ iff
every instantiation of $\T$ defines a well-founded relation.
\end{definition}

\begin{example}\label{ex:rt-pr}
We call the following template with parameters $D = \{ \delta \}$
and affine-linear function symbols $F = \{ f \}$ the
\emph{PR ranking template}~\cite{PR04}.
\begin{align}
\delta > 0 \;\land\; f(x) > 0 \;\land\; f(x') < f(x) - \delta
\label{eq:rt-pr}
\end{align}
In the remainder of this section,
we introduce a formalism that allows us to show that
every instantiation of the PR ranking template defines a well-founded relation.
Let us now check the additional syntactic requirements for \eqref{eq:rt-pr}
to be a linear ranking template:
\renewcommand{\qed}{} 
\begin{align*}
\delta > 0 \;&\equiv\;
  \big( 0 \cdot f(x) + 0 \cdot f(x') \big)
  + 1 \cdot \delta > 0 \\
f(x) > 0 \;&\equiv\;
  \big( 1 \cdot f(x) + 0 \cdot f(x') \big)
  + 0 \cdot \delta > 0 \\
f(x') < f(x) - \delta \;&\equiv\;
  \big( 1 \cdot f(x) + (- 1) \cdot f(x') \big)
  + (-1) \cdot \delta > 0
\tag*{$\Diamond$} 
\end{align*}
\end{example}

The next lemma states that we can prove termination of a given linear loop program
by checking that this program's transition relation
is included in an instantiation of a linear ranking template.

\begin{lemma}[Termination]
\label{lem:termination}
Let $\loopt$ be a linear loop program and
let $\T$ be a linear ranking template
with parameters $D$ and affine-linear function symbols $F$.
If there is an assignment $\nu$ to $D$ and $F$ such that the formula
\begin{align}
\forall x, x'.\;
\big( \loopt(x, x') \rightarrow \nu(\T)(x, x') \big)
\label{eq:rt}
\end{align}
is valid, then the program $\loopt$ terminates.
\end{lemma}
\begin{proof}
By definition, $\nu(\T)$ is a well-founded relation and
\eqref{eq:rt} is valid iff
the relation $\loopt$ is a subset of $\nu(\T)$.
Thus $\loopt$ must be well-founded.
\end{proof}

In order to establish that a relation template
which is conforming to the syntactic requirements
is indeed a ranking template,
we have to show that each instantiation of the relation template is well-founded.
According to the following lemma,
we can do this by showing that each assignment to $D$ and $F$
gives rise to a ranking function.
A similar argument was given in \cite{Ben-Amram09};
we provide a significantly shortened proof
by use of the \hyperref[thm:recursion]{Recursion Theorem},
along the lines of \cite[Ex.\ 6.12]{Jech06}.

\begin{definition}[Ranking Function]
\label{def:rf}
Given a binary relation $R$ over a set $\Sigma$,
a function $\rho$ from $\Sigma$ to some ordinal $\alpha$ is
a \emph{ranking function} for $R$ iff
for all $x,x'\in \Sigma$ the following implication holds.
\[
(x, x') \in R \;\Longrightarrow\; \rho(x) > \rho(x')
\]
\end{definition}

\begin{lemma}[Existence of Ranking Functions]
\label{lem:rf}
A binary relation $R$ is well-founded if and only if there exists a ranking function for $R$.
\end{lemma}
\begin{proof}
Let $\rho$ be a ranking function for $R$.
The image of a sequence decreasing with respect to $R$ under $\rho$ is a
strictly decreasing ordinal sequence.
Because the ordinals are well-ordered, this sequence cannot be infinite.

Conversely, the graph $G = (\Sigma, R)$ with vertices $\Sigma$ and edges $R$ is acyclic by assumption.
Hence the function $\rho$ that assigns to every element of $\Sigma$ an ordinal number such that
$\rho(x) = \sup\, \{ \rho(x') + 1 \mid (x, x') \in R \}$
is well-defined and exists due to the \hyperref[thm:recursion]{Recursion Theorem}.
\end{proof}

\begin{example}\label{ex:running2}
Consider the terminating linear loop program $\loopt$ from \autoref{ex:running1}.
A ranking function for $\loopt$ is $\rho: \mathbb{R}^2 \to \omega$, defined as follows.
\begin{align*}
\rho(q, y) =
\begin{cases}
\lceil q \rceil, &\text{if } q > 0,
\text{ and} \\
0 & \text{otherwise,}
\end{cases}
\end{align*}
where $\lceil \cdot \rceil$ denotes the ceiling function
that assigns to every real number $r$
the smallest natural number that is larger or equal to $r$.
Since we consider the natural numbers to be a subset of the ordinals,
the ranking function $\rho$ is well-defined.
\end{example}

We use assignments to a template's
{parameters and affine-linear function symbols} to construct a ranking function.
These functions are real-valued and we transform them into
ordinal-valued functions as follows.

\begin{definition}[Ordinal Ranking Equivalent]
\label{def:ordinal-ranking}
Given an affine-linear function $f$ and a real number $\delta > 0$ called
the \emph{step size}, we define the \emph{ordinal ranking equivalent} of $f$ as
\begin{align*}
\widehat{f}(x) =
\begin{cases}
\left\lceil \frac{f(x)}{\delta} \right\rceil, &\text{if } f(x) > 0,
\text{ and} \\
0 & \text{otherwise.}
\end{cases}
\end{align*}
\end{definition}

Our notation does not explicitly refer to $\delta$ to increase readability.
In our presentation the step size $\delta$ is always clear from
the context in which an ordinal ranking equivalent $\widehat{f}$ is used.

\begin{example}\label{ex:running3}
Consider the linear loop program $\loopt(x, x')$ from \autoref{ex:running1}.
For $\delta = 1/2$ and $f(q) = q + 1$,
the ordinal ranking equivalent of $f$ with step size $\delta$ is
\begin{align*}
\widehat{f}(q, y) =
\begin{cases}
\lceil 2(q + 1) \rceil, & \text{if } q + 1 > 0, \text{ and} \\
0 & \text{otherwise.}
\end{cases}
\end{align*}
\end{example}

The assignment from \autoref{ex:running3} to $\delta$ and $f$ makes the implication \eqref{eq:rt} valid.
In order to invoke \autoref{lem:termination} to show that the linear loop program given in \autoref{ex:running1} terminates,
we need to prove that the PR ranking template is a linear ranking template.
We use the following technical lemma.

\begin{lemma}[Well-Foundedness of Ordinal Ranking Equivalents]
\label{lem:ordinal-ranking}
Let $f$ be an affine-linear function of step size $\delta > 0$ and
let $x$ and $x'$ be two states.
If $f(x) > 0$ and $f(x) - f(x') > \delta$,
then $\widehat{f}(x) > 0$ and $\widehat{f}(x) > \widehat{f}(x')$.
\end{lemma}
\begin{proof}
From $f(x) > 0$ follows that $\widehat{f}(x) > 0$.
Therefore $\widehat{f}(x) > \widehat{f}(x')$ in the case $\widehat{f}(x') = 0$.
For $\widehat{f}(x') > 0$, we use the fact that $f(x) - f(x') > \delta$
to conclude that $f(x) / \delta - f(x') / \delta > 1$ and
hence $\widehat{f}(x') > \widehat{f}(x)$.
\end{proof}

\begin{corollary}\label{cor:rt-pr}
The PR ranking template is a linear ranking template.
\end{corollary}
\begin{proof}
Any assignment $\nu$ to $\delta$ and $f$
satisfies the requirements of \autoref{lem:ordinal-ranking}.
Consequently, $\widehat{f}$ is a ranking function for $\nu(\T)$,
and by \autoref{lem:rf} this implies that $\nu(\T)$ is well-founded.
\end{proof}

The goal of this paper is to use linear ranking templates
to prove the termination of linear lasso programs,
as exposed in \autoref{lem:termination}.
We defer the explanation how the $\exists\forall$-formula \eqref{eq:rt}
can be transformed so that it is easier to solve
to \autoref{sec:synthesizing-ranking-functions}.
The next two sections focus on additional examples for linear ranking templates.


\section{Examples of Ranking Templates}
\label{sec:ranking-templates}

\subsection{The Multiphase Ranking Template}
\label{ssec:rt-multiphase}

The multiphase ranking template targets
programs that go through a finite number of phases in their execution.
Each phase is ranked with an affine-linear function and
the phase is considered to be completed once this function becomes non-positive.

\begin{example}\label{ex:2phase-1}
Consider the linear loop program from \autoref{fig:multiphase-introduction}
on page \pageref{fig:multiphase-introduction}.
Every execution can be partitioned into two phases:
first $y$ increases until it is positive and
then $q$ decreases until the loop condition $q > 0$ is violated.
Depending on the initial values of $y$ and $q$,
one or more phases might be skipped altogether.
\end{example}

\begin{definition}[Multiphase Ranking Template]
\label{def:rt-multiphase}
We define the \emph{$k$-phase ranking template} with
parameters $D = \{ \delta_1, \ldots, \delta_k \}$ and
affine-linear function symbols $F = \{ f_1, \ldots, f_k \}$ as follows.
\begin{align}
&\bigwedge_{i=1}^k \delta_i > 0 \label{eq:rt-multiphase1} \\
\land\; &\bigvee_{i=1}^k f_i(x) > 0 \label{eq:rt-multiphase2} \\
\land\; &f_1(x') < f_1(x) - \delta_1 \label{eq:rt-multiphase3} \\
\land\; &\bigwedge_{i=2}^k \Big( f_i(x') < f_i(x) - \delta_i
  \;\lor\; f_{i-1}(x) > 0 \Big) \label{eq:rt-multiphase4}
\end{align}
\end{definition}

We say that the multiphase ranking function given by
an assignment to $f_1, \ldots, f_k$ and $\delta_1, \ldots, \delta_k$
\emph{is in phase~$i$}
iff $f_i(x) > 0$ and $f_j(x) \leq 0$ for all $j < i$.
The condition \eqref{eq:rt-multiphase2} states that there is always some $i$ such that
the multiphase ranking function is in phase~$i$.
Conditions \eqref{eq:rt-multiphase3} and \eqref{eq:rt-multiphase4} state that
if we are in a phase $\geq i$,
then $f_i$ has to be decreasing by at least $\delta_i > 0$.
Thus we start in phase~$1$ and transition through the phases $2, \ldots, k$,
possibly skipping some or all of them.

Two special cases of the multiphase ranking template
have been discussed previously in the literature:
the $1$-phase ranking template,
because it coincides with the PR ranking template, and
the $2$-phase ranking template~\cite{BM13}.

Moreover,
multiphase ranking functions are related to
\emph{eventually negative expressions}
introduced by Bradley, Manna, and Sipma~\cite{BMS05polyrank}.
However, in contrast to our approach,
they require a template tree that specifies in detail
how each loop disjunct interacts with each phase.

\begin{example}\label{ex:2-phase}
Consider the program from \autoref{fig:multiphase-introduction}
on page \pageref{fig:multiphase-introduction}.
The assignment
\[
       f_1(q, y) = 1 - y,
\qquad f_2(q, y) = q + 1,
\qquad \delta_1 = \delta_2 = \tfrac{1}{2}
\]
yields a $2$-phase ranking function for this program.
This program is in phase~$1$ iff $y < 1$
and it is in phase~$2$ iff $y \geq 1$ and $q > -1$.
\end{example}

\begin{theorem}\label{thm:rt-multiphase}
The $k$-phase ranking template is a linear ranking template.
\end{theorem}
\begin{proof}
The $k$-phase ranking template conforms to
the linear ranking template's syntactic requirements.
Let $\nu$ be an assignment to the parameters $D$
and the affine-linear function symbols $F$
of the $k$-phase ranking template $\T_{k\text{-phase}}$.
Consider the following ranking function with codomain $\omega \cdot k$.
\begin{equation}\label{eq:rt-multiphase-rf}
\rho(x) :=
\begin{cases}
\omega \cdot (k - i) + \widehat{f_i}(x) & \text{if }
f_j(x) \leq 0 \text{ for all } j < i \text{ and } f_i(x) > 0, \\
0 & \text{otherwise.}
\end{cases}
\end{equation}
Let $(x, x') \in \nu(\T_{k\text{-phase}})$.
By \autoref{lem:rf}, we need to show that $\rho(x') < \rho(x)$.
From \eqref{eq:rt-multiphase2} follows that $\rho(x) > 0$.
Moreover, there is an $i$ such that $f_i(x) > 0$ and $f_j(x) \leq 0$ for all $j < i$.
By \eqref{eq:rt-multiphase3} and \eqref{eq:rt-multiphase4}, we obtain
$f_j(x') \leq 0$ for all $j < i$, because
$f_j(x') < f_j(x) - \delta_j \leq 0 - \delta_j < 0$, since
$f_\ell(x) \leq 0$ for all $\ell < j$.

If $f_i(x') \leq 0$, then $\rho(x') \leq \omega \cdot (k - i)
< \omega \cdot (k - i) + \widehat{f_i}(x) = \rho(x)$.
Otherwise, $f_i(x') > 0$ and from \eqref{eq:rt-multiphase4} follows $f_i(x') < f_i(x) - \delta_i$.
By \autoref{lem:ordinal-ranking},
$\widehat{f_i}(x) > \widehat{f_i}(x')$ for the ordinal ranking equivalent of $f_i$ with step size $\delta_i$.
Hence
\[
\rho(x') = \omega \cdot (k - i) + \widehat{f_i}(x')
  < \omega \cdot (k - i) + \widehat{f_i}(x)
  = \rho(x).
\]
Therefore \autoref{lem:rf} implies that
$\nu(\T_{k\text{-phase}})$ is well-founded.
\end{proof}

Does each terminating linear loop program
have a multiphase ranking function
if we restrict ourselves to conjunctive linear loop programs?
The following theorem gives a negative answer to this question.

\begin{theorem}\label{thm:conjunctive-loop-no-multiphase}
The following terminating conjunctive linear loop program
does not have a multiphase ranking function.
\begin{equation}\label{eq:conjunctive-loop-no-multiphase}
a > b \;\land\; b > 1
\;\land\; a' = 2a \;\land\; b' = 3b
\end{equation}
\end{theorem}
\begin{proof}
The variables $a$ and $b$ are positive and grow exponentially,
but $b$ grows faster than $a$.
For any input, $b$ eventually becomes larger than $a$ and then
the loop program terminates.

Assume the loop program \eqref{eq:conjunctive-loop-no-multiphase}
has a multiphase ranking function.
Then there are $\alpha_1, \beta_1, \gamma_1, \ldots, \alpha_k, \beta_k, \gamma_k \in \mathbb{R}$ such that
$f_i(a, b) = \alpha_i a + \beta_i b + \gamma_i$ for all $1 \leq i \leq k$.
Choose
\[
b := \max \{ 2 \} \cup \left\{ \tfrac{\gamma_i}{\beta_i}, \tfrac{-\gamma_i}{\beta_i}
       \mid \beta_i \neq 0 \right\},
\quad\quad
a := \max \{ b + 1 \} \cup \left\{ -2\beta_i b, \tfrac{-\gamma_i}{\alpha_i}
       \mid \alpha_i \neq 0 \right\}.
\]
As maxima over a finite nonempty set, $a, b \in \mathbb{R}$ exist uniquely
and we have $a > b > 1$ by construction.
By setting $a' = 2a$ and $b' = 3b$, we get $(a, b, 2a, 3b) \in \loopt$.
Let $j$ be the smallest index such that $f_j(a, b) > 0$,
which exists due to \eqref{eq:rt-multiphase2}.
According to \eqref{eq:rt-multiphase1} we obtain $\delta_j > 0$ and
since $j$ is minimal, we get $f_j(a', b') < f_j(a, b)$
from \eqref{eq:rt-multiphase4} (\eqref{eq:rt-multiphase3} in case $j = 1$).
Hence
\begin{equation}\label{eq:f_j-decreasing}
  0
> f_j(a', b') - f_j(a, b)
= 2\alpha_j a + 3\beta_j b - \alpha_j a - \beta_j b
= \alpha_j a + 2\beta_j b.
\end{equation}
We do an exhaustive case analysis over $\alpha_j$ and $\beta_j$,
and show that all cases yield contradictions.
Thus our assumption that there is a multiphase ranking function must be false.
\begin{enumerate}[label=(\roman*)]
\item $\alpha_j > 0$:
From \eqref{eq:f_j-decreasing} and $a \geq -2\beta_j b$ we get
\[
     0
>    \alpha_j a + 2\beta_j b
\geq -2\beta_i b + 2\beta_j b
=    0.
\]
\item $\beta_j > 0$:
From \eqref{eq:f_j-decreasing} and $b \geq \gamma_j / \beta_j$ we get
\[
     0
>    \alpha_j a + 2\beta_j b
\geq \alpha_j a + \beta_j b + \beta_j \tfrac{\gamma_j}{\beta_j}
=    f_j(a, b)
>    0.
\]
\item $\alpha_j = \beta_j = 0$:
From \eqref{eq:f_j-decreasing} we get
$
  0
> \alpha_j a + 2\beta_j b
= 0
$.
\item $\alpha_j < 0$ and $\beta_j \leq 0$:
From $a \geq -\gamma_j / \alpha_j$ we get
\[
     0
<    f_j(a, b)
=    \alpha_j a + \beta_j b + \gamma_j
\leq \alpha_j a + \gamma_j
\leq \alpha_j \tfrac{-\gamma_j}{\alpha_j} + \gamma_j
=    0.
\]
\item $\beta_j < 0$ and $\alpha_j \leq 0$:
From $b \geq -\gamma_j / \beta_j$ we get
\[
     0
<    f_j(a, b)
=    \alpha_j a + \beta_j b + \gamma_j
\leq \beta_j b + \gamma_j
\leq \beta_j \tfrac{-\gamma_j}{\beta_j} + \gamma_j
=    0.
\qedhere
\]
\end{enumerate}
\noqed
\end{proof}

\begin{example}\label{ex:multiphase-complexity}
Recall that we allowed our linear loop programs
to have nondeterministic variable assignments.
Because of this,
the existence of a multiphase ranking function
does not imply an upper bound on the execution time of the program.
Consider the following linear loop program.
\begin{align*}
  &(q > 0 \;\land\; y > 0 \;\land\; y' = 0) \\
\lor\; &(q > 0 \;\land\; y \leq 0 \;\land\; y' = y - 1 \;\land\; q' = q - 1)
\end{align*}
For a given input with $y > 0$,
we cannot give an upper bound on the execution time:
after the first loop execution, $y$ is set to $0$ and
$q$ is set to \emph{some arbitrary value}, as no restriction to $q'$ applies in the first disjunct.
In particular, this value does not depend on the input.
The remainder of the loop execution then takes $\lceil q \rceil$ iterations to terminate.

However, we can prove the program's termination with
the $2$-phase ranking function constructed from
$f_1(q, y) = y$ and $f_2(q, y) = q$.
\end{example}

\subsection{The Nested Ranking Template}
\label{ssec:rt-nested}

Like the multiphase ranking template,
the nested ranking template targets programs
that go through a fixed number of phases in their execution.
Again, each phase has an affine-linear ranking function,
but this affine-linear function cannot increase by more than
the value of the previous phase's affine-linear function.
Thus once the previous phase is finished,
its value starts decreasing.

\begin{definition}[Nested Ranking Template]
\label{def:rt-nested}
We define the \emph{$k$-nested ranking template} with
parameters $D = \{ \delta \}$ and
affine-linear function symbols $F = \{ f_1, \ldots, f_k \}$ as follows.
\begin{align}
&\delta > 0 \label{eq:rt-nested1} \\
\land\; &f_k(x) > 0 \label{eq:rt-nested2} \\
\land\; &f_1(x') < f_1(x) - \delta \label{eq:rt-nested3} \\
\land\; &\bigwedge_{i=2}^k f_i(x') < f_i(x) + f_{i-1}(x) \label{eq:rt-nested4}
\end{align}
\end{definition}

\begin{example}\label{ex:2-nested}
Consider the program from \autoref{fig:multiphase-introduction}.
In \autoref{ex:2-phase} we gave an assignment for the $2$-phase ranking template.
We can use almost the same assignment
\[
       f_1(q, y) = 1 - y,
\qquad f_2(q, y) = q + 1,
\qquad \delta = \tfrac{1}{2}
\]
to get a $2$-nested ranking function for this program.
\end{example}

The following lemma states that nested ranking functions
are a special case of multiphase ranking functions.

\begin{lemma}[Nested Ranking Template $\subseteq$ Multiphase Ranking Template]
\label{lem:nested-and-multiphase}
For every assignment $\nu$
to the $k$-phase ranking template $\T_{k\text{-phase}}$
there is an assignment $\nu'$
to the $k$-nested ranking template $\T_{k\text{-nested}}$
such that $\nu'(\T_{k\text{-nested}}) \subseteq \nu(\T_{k\text{-phase}})$
\end{lemma}
\begin{proof}
For a given $\nu$, we choose
\[
\nu'(\delta) := \nu(\delta_1),
\quad
\nu'(f_i) := \nu(f_i) - \nu(\delta_{i+1})
\quad
\nu'(f_k) := \nu(f_k).
\]
We show $\nu'(\T_{k\text{-nested}}) \subseteq \nu(\T_{k\text{-phase}})$
by showing that each of
\eqref{eq:rt-nested1}, \eqref{eq:rt-nested2}, \eqref{eq:rt-nested3},
and \eqref{eq:rt-nested4}
with assignment $\nu'$ implies
\eqref{eq:rt-multiphase1}, \eqref{eq:rt-multiphase2}, \eqref{eq:rt-multiphase3},
and \eqref{eq:rt-multiphase4} with assignment $\nu$, respectively.
This is immediate for the first three lines.
For \eqref{eq:rt-nested4} $\rightarrow$ \eqref{eq:rt-multiphase4}
let $(x, x')$ and $i > 1$ be given and assume $\nu(f_{i-1})(x) \leq 0$.
We get
\begin{align*}
      \nu(f_i)(x')
&=    \nu'(f_i)(x') + \nu(\delta_{i+1}) \\
&\leq \nu'(f_i)(x') + \nu(\delta_{i+1}) - \nu(f_{i-1})(x) \\
&<    \nu'(f_i)(x) + \nu'(f_{i-1})(x) + \nu(\delta_{i+1}) - \nu(f_{i-1})(x) \\
&=    \nu(f_i)(x) + \nu(f_{i-1})(x) - \nu(\delta_i) - \nu(f_{i-1})(x) \\
&=    \nu(f_i)(x) - \nu(\delta_i),
\end{align*}
with $\nu(\delta_{k+1}) := 0$ for notational convenience.
\end{proof}

\begin{theorem}\label{thm:rt-nested}
The $k$-nested ranking template is a linear ranking template.
\end{theorem}
\begin{proof}
Follows from \autoref{thm:rt-multiphase}
and \autoref{lem:nested-and-multiphase}.
\end{proof}

If a nested ranking function
is just a special case of a multiphase ranking function,
why are we considering it separately?
The advantage of the nested template is that
it does not contain any disjunctions.
Thus, the generated constraint can be solved
using only linear constraint solving~\cite[Ch.\ 6]{Leike13}.
In our experiments,
many programs that have a multiphase ranking function
also have a nested ranking function.
Hence it is a viable and faster alternative to the multiphase template in practice.

\begin{example}\label{ex:multiphase-but-no-nested}
The multiphase template is strictly more powerful than the nested template;
consider the following loop program $\loopt$.
\[
          (q > 0 \;\lor\; y > 0)
\;\land\; y' = y - 1
\;\land\; q' \leq q
\;\land\; (y \leq 0 \rightarrow q' = q - 1)
\]
This program has the $2$-phase ranking function constructed from
$f_1(q, y) = y$ and $f_2(q, y) = q$.
Since there are no upper or lower bounds on $q$ and $y$,
only the constant function $f_i(q, y) = \gamma_i$
can be positive for all $q$, $y$.
By \eqref{eq:rt-nested4} if $f_i$ is constant, then $f_{i-1}$ is positive.
By induction we get that $f_1$ must be constant,
a contradiction to \eqref{eq:rt-nested3}.
\end{example}

Despite their simplicity,
nested ranking templates are already quite powerful.
We give two nontrivial examples below
that each have a nested ranking function.
These ranking functions were found automatically.

\begin{example}[Rotation53]\label{ex:spirals}
Consider the following conjunctive linear loop program.
\[
q > 0
\land q' = q + a - 1
\land a' = \tfrac{3}{5} a - \tfrac{4}{5} b
\land b' = \tfrac{4}{5} a + \tfrac{3}{5} b
\]
During the execution of this loop,
the vector $(a, b)$ is rotated around $0$
by the irrational angle $\arccos(3/5) \approx 53.13$ degrees.
In the long run,
the contribution of $a$ to $q$ cancels out,
and $q$ decreases on average, hence the program terminates.
This program has a $3$-nested ranking function constructed from
the affine-linear functions
\renewcommand{\qed}{} 
\[
f_1(q, a, b) = 2q + a - 2b,
\quad\quad
f_2(q, a, b) = 4q + 5a,
\quad\text{and}\quad
f_3(q, a, b) = 5q.
\tag*{$\Diamond$} 
\]
\end{example}

\begin{example}[Crazy Spirals]\label{ex:crazy-spirals}
We can also take the previous example to more extremes;
consider the following conjunctive linear loop program.
\[
      q > 0
\land q' = q + a - 1
\land a' = 3a - 5b + c
\land b' = 12a + 3b
\land c' = 3c - 4d
\land d' = 4c + 3d
\]
During program execution,
the vector $(c, d)$ moves on an outward spiral centered at $0$;
the vector $(a, b)$ does the same except that it is offset by $c$.
On average, the contribution from these spirals to $q$ cancel out,
so $q$ decreases on average.
This program has a $7$-nested ranking function.
\end{example}

\subsection{The Piecewise Ranking Template}
\label{ssec:rt-pw}

The piecewise ranking template formalizes a ranking function that is defined piecewise using affine-linear
predicates to discriminate the pieces.

\begin{definition}[Piecewise Ranking Template]
\label{def:rt-pw}
We define the \emph{$k$-piece ranking template} with
parameters $D = \{ \delta \}$ and
affine-linear function symbols $F = \{ f_1, \ldots, f_k, g_1, \ldots, g_k \}$
as follows.
\begin{align}
&\delta > 0 \label{eq:rt-pw1} \\
\land\; &\bigwedge_{i=1}^k \bigwedge_{j=1}^k \Big( g_i(x) < 0
  \;\lor\; g_j(x') < 0 \;\lor\; f_j(x') < f_i(x) - \delta \Big) \label{eq:rt-pw2} \\
\land\; &\bigwedge_{i=1}^k f_i(x) > 0 \label{eq:rt-pw3} \\
\land\; &\bigvee_{i=1}^k g_i(x) \geq 0 \label{eq:rt-pw4}
\end{align}
\end{definition}
\noindent We call the affine-linear function symbols $\{ g_i \mid 1 \leq i \leq k \}$ \emph{discriminators} and
the affine-linear function symbols $\{ f_i \mid 1 \leq i \leq k \}$ \emph{ranking pieces}.

The disjunction \eqref{eq:rt-pw4} states that the discriminators cover all states;
in other words, the piecewise defined ranking function is a total function.
Given the $k$ different pieces $f_1, \ldots, f_k$ and a state $x$,
we use $f_i$ as a ranking function only if $g_i(x) \geq 0$ holds.
This choice need not be unambiguous;
the discriminators may overlap.
If they do, we can use any one of their ranking pieces.
According to \eqref{eq:rt-pw3}, all ranking pieces are positive-valued and
by \eqref{eq:rt-pw2}, ranking piece transitions are well-defined:
the rank of the new state is always less than
the rank of any of the ranking pieces assigned to the old state.

\begin{example}\label{ex:2-piece}
Consider the following linear loop program.
\begin{align*}
&(q > 0 \;\land\; p > 0 \;\land\; q < p \;\land\; q' = q - 1) \\
\lor\; &(q > 0 \;\land\; p > 0 \;\land\; p < q \;\land\; p' = p - 1)
\end{align*}
In every loop iteration, the minimum of $p$ and $q$ is decreased by $1$ until it becomes negative.
Thus, this program is ranked by the 2-piece ranking function constructed from
the ranking pieces $f_1(p, q) = p$ and $f_2(p, q) = q$
with step size $\delta = 1/2$ and
discriminators $g_1(p, q) = q - p$ and $g_2(p, q) = p - q$.
Moreover, this program does not have a multiphase or lexicographic ranking function:
both $p$ and $q$ may increase without bound during program execution due to non-determinism and
the number of switches between $p$ and $q$ being the minimum value is also unbounded.
\end{example}

\begin{theorem}\label{thm:rt-pw}
The $k$-piece ranking template is a linear ranking template.
\end{theorem}
\begin{proof}
The $k$-piece ranking template conforms to
the linear ranking template's syntactic requirements.
Let $\nu$ be an assignment to the parameter $\delta$
and the affine-linear function symbols $F$
of the $k$-piece template $\T_{k\text{-piece}}$ be given.
Consider the following ranking function with codomain $\omega$.
\begin{equation}\label{eq:rf-pw}
\rho(x) := \max \big\{ \widehat{f_i}(x) \mid g_i(x) \geq 0 \big\}
\end{equation}
The function $\rho$ is well-defined,
because the set $\{ \widehat{f_i}(x) \mid g_i(x) \geq 0 \}$ is not empty
according to \eqref{eq:rt-pw4}.
Let $(x, x') \in \nu(\T_{k\text{-piece}})$ and
let $i$ and $j$ be indices such that
$\rho(x) = \widehat{f_i}(x)$ and $\rho(x') = \widehat{f_j}(x')$.
By the definition of $\rho$,
we have that $g_i(x) \geq 0$ and $g_j(x') \geq 0$, and
\eqref{eq:rt-pw2} thus implies $f_j(x') < f_i(x) - \delta$.
Using \eqref{eq:rt-pw3},
we prove analogously to \autoref{lem:ordinal-ranking},
that this entails $\widehat{f_j}(x') < \widehat{f_i}(x)$
and therefore $\rho(x') < \rho(x)$.
\autoref{lem:rf} now implies that
$\nu(\T_{k\text{-piece}})$ is well-founded.
\end{proof}

\subsection{The Lexicographic Ranking Template}
\label{ssec:rt-lex}

Lexicographic ranking functions consist of lexicographically ordered components of affine-linear functions.
A state is mapped to a tuple of values such that the loop transition leads to a decrease with respect to the lexicographic ordering for this tuple.
Therefore no function may increase unless a function of a lower index decreases.
Additionally, at every step, there must be at least one function that decreases.

There are different definitions of lexicographic ranking functions
in circulation~\cite{ADFG10,BAG13,BMS05linrank};
a comparison can be found in \cite[Sec.\ 2.4]{BAG13}.
Each of these definitions for lexicographic linear ranking functions
can be formalized using linear ranking templates.
Here we are following the definition of \cite{ADFG10}.
This definition is the weakest,
but for the other definitions
the ranking template has an exponentially larger CNF,
and hence our method performs comparatively poorly on them.

\begin{definition}[Lexicographic Ranking Template]\label{def:rt-lex}
We define the \emph{$k$-lexi\-co\-gra\-phic ranking template}
with parameters $D = \{ \delta_1, \ldots, \delta_k \}$ and
affine-linear function symbols $F = \{ f_1, \ldots, f_k \}$ as follows.
\begin{align}
&\bigwedge_{i=1}^{k} \delta_i > 0 \label{eq:rt-lex1} \\
\land\; &\bigwedge_{i=1}^k f_i(x) > 0 \label{eq:rt-lex2} \\
\land\; &\bigwedge_{i=1}^{k-1} \Big( f_i(x') \leq f_i(x)
  \;\lor\; \bigvee_{j=1}^{i-1} f_j(x') < f_j(x) - \delta_j \Big) \label{eq:rt-lex3} \\
\land\; &\bigvee_{i=1}^k f_i(x') < f_i(x) - \delta_i \label{eq:rt-lex4}
\end{align}
\end{definition}
\noindent The conjunction \eqref{eq:rt-lex2} establishes that all lexicographic components $f_1, \ldots, f_k$ have positive values.
In every step, at least one component must decrease according to \eqref{eq:rt-lex4}.
From \eqref{eq:rt-lex3} follows that
all functions corresponding to components of larger index
than the decreasing function may increase.

\begin{example}\label{ex:rt-lex}
Consider the following linear loop program.
\begin{align*}
&(a > 0 \land b > 5 \land a' = a \land b' = b - 1) \\
\lor\;
&(a > 0 \land b > 0 \land a' = a - 1 \land b' > 0)
\end{align*}
When taking the first disjunct, $b$ decreases until it becomes $\leq 5$.
Hence we take the second disjunct eventually, decreasing $a$.
Because $a$ does not increase when taking the first disjunct,
we can only take the second disjunct finitely many times.
Since the second disjunct is always taken eventually, the program terminates.

This is proved by the $2$-lexicographic ranking function constructed from
the components $f_1(a, b) = a$ and $f_2(a, b) = b$.
\end{example}

Note that the program from \autoref{ex:rt-lex} does not have a multiphase ranking function
and the program from \autoref{fig:multiphase-introduction} on page \pageref{fig:multiphase-introduction}
does not have a lexicographic ranking function.
Thus the multiphase ranking template and the lexicographic ranking template
are incomparable in expressive power.

\begin{theorem}\label{thm:rt-lex}
The $k$-lexicographic ranking template is a linear ranking template.
\end{theorem}
\begin{proof}
The $k$-lexicographic ranking template conforms to
the linear ranking template's syntactic requirements.
Let $\nu$ be an assignment to the parameters $D$
and the affine-linear function symbols $F$
of the $k$-lexicographic template $\T_{k\text{-lex}}$.
Consider the following ranking function with codomain $\omega^k$.
\begin{equation}\label{eq:rf-lex}
\rho(x) := \sum_{i=1}^k \omega^{k-i} \cdot \widehat{f_i}(x)
\end{equation}
Let $(x, x') \in \nu(\T_{k\text{-lex}})$.
From \eqref{eq:rt-lex2} follows $f_j(x) > 0$ for all $j$, so $\rho(x) > 0$.
By \eqref{eq:rt-lex4} and \autoref{lem:ordinal-ranking},
there is a minimal $i$ such that $\widehat{f_i}(x') < \widehat{f_i}(x)$.
According to \eqref{eq:rt-lex3},
we have $\widehat{f_1}(x') \leq \widehat{f_1}(x)$ and hence
inductively $\widehat{f_j}(x') \leq \widehat{f_j}(x)$ for all $j < i$,
since $i$ was minimal.
\begin{align*}
\rho(x') &= \sum_{j=1}^k \omega^{k-j} \cdot \widehat{f_j}(x')
  \leq \sum_{j=1}^{i-1} \omega^{k-j} \cdot \widehat{f_j}(x)
    + \sum_{j=i}^k \omega^{k-j} \cdot \widehat{f_j}(x') \\
  &< \sum_{j=1}^{i-1} \omega^{k-j} \cdot \widehat{f_j}(x)
    + \omega^{k-i} \cdot \widehat{f_i}(x)
  \leq \rho(x)
\end{align*}
Therefore \autoref{lem:rf} implies that
$\nu(\T_{k\text{-lex}})$ is well-founded.
\end{proof}

\subsection{The Parallel Ranking Template}
\label{ssec:rt-parallel}

The parallel ranking template targets programs
that do multiple tasks in parallel
where progress on each task can be nondeterministic.
These tasks have no predetermined order of execution.
We assume that each task can be ranked by an affine-linear ranking function.

\begin{definition}[Parallel Ranking Template]\label{def:rt-parallel}
We define the \emph{$k$-parallel ranking template} with
parameters $D = \{ \delta_1, \ldots, \delta_k \}$ and
affine-linear function symbols $F = \{ f_1, \ldots, f_k \}$ as follows.
\begin{align}
        &\bigwedge_{i=1}^k \delta_i > 0 \label{eq:rt-parallel1} \\
\land\; &\bigwedge_{i=1}^k f_i(x') \leq f_i(x) \label{eq:rt-parallel2} \\
\land\; &\bigvee_{i=1}^k \Big( f_i(x) > 0
           \;\land\; f_i(x') < f_i(x) - \delta_i \Big) \label{eq:rt-parallel3}
\end{align}
\end{definition}

\noindent The ranking functions $f_1, \ldots, f_k$ correspond to $k$ different tasks.
The conjunction \eqref{eq:rt-parallel2} states that
none of the ranking functions may increase at any point.
Moreover, \eqref{eq:rt-parallel3} states that with every transition,
at least one task has to make progress and end in a finite number of steps.
Note that \eqref{eq:rt-parallel3} is not given in conjunctive normal form (CNF).
When transformed in CNF, the number of conjuncts blows up exponentially.

\begin{example}\label{ex:parallel}
Consider the following linear loop program.
\begin{align*}
       &(a > 0 \;\land\; a' = a - 1 \;\land\; b' = b) \\
\lor\; &(b > 0 \;\land\; b' = b - 1 \;\land\; a' = a)
\end{align*}
This programs performs two tasks nondeterministically in parallel:
the first task takes $a$ iterations and the second task takes $b$ iterations;
the program terminates once both tasks have been completed.
The $2$-parallel ranking function constructed from
$f_1(a, b) = a$ and $f_2(a, b) = b$ proves this program terminating.
\end{example}

\begin{theorem}\label{thm:rt-parallel}
The $k$-parallel ranking template is a linear ranking template.
\end{theorem}
\begin{proof}
The $k$-parallel ranking template conforms to
the linear ranking template's syntactic requirements.
Let $\nu$ be an assignment to the parameters $D$ and
the affine-linear function symbols $F$
of the $k$-parallel template $\T_{k\text{-parallel}}$.
Consider the following ranking function with codomain $\omega$.
\begin{equation}\label{eq:rf-parallel}
\rho(x) := \sum_{i=1}^k \widehat{f_i}(x)
\end{equation}
Let $(x, x') \in \nu(\T_{k\text{-parallel}})$.
From \eqref{eq:rt-parallel2} follows $f_i(x') \leq f_i(x)$ for all $i$,
so $\widehat{f_i}(x') \leq \widehat{f_i}(x)$.
By \eqref{eq:rt-parallel3}, there is a $j$ such that
$f_j(x) > 0$ and $f_j(x') < f_j(x) - \delta_j$.
Therefore we have $\widehat{f_j}(x') < \widehat{f_j}(x)$
by \autoref{lem:ordinal-ranking}.
Hence
\[
     \rho(x')
=    \sum_{i=1}^k \widehat{f_i}(x')
=    \widehat{f_j}(x') + \sum_{i \neq j} \widehat{f_i}(x')
\leq \widehat{f_j}(x') + \sum_{i \neq j} \widehat{f_i}(x)
<    \widehat{f_j}(x)  + \sum_{i \neq j} \widehat{f_i}(x)
=    \rho(x).
\]
Now \autoref{lem:rf} implies that
$\nu(\T_{k\text{-parallel}})$ is well-founded.
\end{proof}


\section{Composition of Templates}
\label{sec:composition-of-templates}

In this section we discuss how more powerful linear ranking templates
can be constructed based on the linear ranking templates
from \autoref{sec:ranking-templates}.
First, we consider a program that is terminating,
but whose termination cannot be proven
using one of the ranking templates presented so far.

\begin{example}\label{ex:multilex}
Consider the following linear loop program.
\begin{align*}
&(q > 0 \land y > 0 \land y' = y - 1 \land q' = q \land x' = x) \\
\lor\;
&(q > 0 \land y \leq 0 \land q' = q - x \land x' = x + 1)
\end{align*}
When executing the first disjunct, $y$ decreases until it becomes negative.
Then we execute the second disjunct:
we increment $x$, set $y$ to some arbitrary value,
and decrement $q$ if $x$ is positive.
If $y$ was reset to some positive value,
the first disjunct is executed again,
but the values of $q$ and $x$ do not change
until the second disjunct is executed.
Eventually, $x$ is positive
and from then on $q$ is decremented until it is nonpositive;
thus the program terminates.
\end{example}

The program's behavior resembles a lexicographic ranking function with
$q$ as the first component and $y$ as the second.
However, $q$ is decremented only after $x$ becomes positive:
the first component needs $2$ phases.
We want a ranking template for a lexicographic ranking function that
has multiphase ranking functions instead of
affine-linear functions in every component.
How can we construct a linear ranking template for such a ranking function?

Observe that all of our linear ranking templates share the following subformulas.
\begin{enumerate}[label=(\roman*)]
\item $f(x') \leq f(x)$
\item $f(x') < f(x) - \delta$
\item $f(x) > 0$
\end{enumerate}
In the context of ranking templates, these formulas have the following meaning.
\begin{enumerate}[label=(\roman*)]
\item The function $f$ is non-increasing.
\item The function $f$ is decreasing.
\item The codomain of the function $f$ is well-founded.
\end{enumerate}
Here $f$ is always an affine-linear function.
The idea of template composition is to replace the subformulas (i--iii)
with subformulas of the same meaning for more powerful functions.
We next define triples of formulas that are suitable substituents
for (i--iii), called \emph{composed template recipes}.
Afterwards, we will use composed template recipes
to build new linear ranking templates.

\begin{definition}[Composed Template Recipe]
\label{def:composed-template-recipe}
A \emph{composed template recipe}
is defined recursively according to the following rules.
\begin{enumerate}[label=(C\arabic*)]
\item The \emph{PR template recipe
	$(\T^\leq_\text{PR}, \T^<_\text{PR}, \T^{>0}_\text{PR})$}
	is a composed template recipe with
\begin{align*}
\T^\leq_\text{PR} &\equiv f(x') \leq f(x) \\
\T^<_\text{PR}    &\equiv f(x') < f(x) - \delta \;\land\; \delta > 0 \\
\T^{>0}_\text{PR} &\equiv f(x) > 0.
\end{align*}
\item The \emph{$k$-piece template recipe
	$(\T^\leq_{k\text{-piece}}, \T^<_{k\text{-piece}}, \T^{>0}_{k\text{-piece}})$}
	is a composed template recipe with
\begin{align*}
        \T^\leq_{k\text{-piece}}
&\equiv \bigvee_{i=1}^k g_i(x) \geq 0
          \;\land\; \bigwedge_{i=1}^k \bigwedge_{j=1}^k \Big(
            g_i(x) < 0 \;\lor\; g_j(x') < 0 \;\lor\; f_j(x') \leq f_i(x)
          \Big) \\
        \T^<_{k\text{-piece}}
&\equiv \delta > 0 \;\land\; \bigvee_{i=1}^k g_i(x) \geq 0
          \;\land\; \bigwedge_{i=1}^k \bigwedge_{j=1}^k \Big(
            g_i(x) < 0 \lor g_j(x') < 0 \lor f_j(x') < f_i(x) - \delta
          \Big) \\
\T^{>0}_{k\text{-piece}} &\equiv \bigwedge_{i=1}^k f_i(x) > 0.
\end{align*}
\item Given $k$ composed template recipes
	$(\T^\leq_1, \T^<_1, \T^{>0}_1), \ldots, (\T^\leq_k, \T^<_k, \T^{>0}_k)$
	which do not share any parameters or affine-linear function symbols,
	we can construct a composed template recipe $(\T^\leq, \T^<, \T^{>0})$
	according to one of the following three composition rules.
\end{enumerate}
\begin{center}
\renewcommand{\arraystretch}{1.3} 
\begin{tabular}{lccc}
Composition rule & $\T^\leq$ & $\T^<$ & $\T^{>0}$ \\
\hline
\hyperref[def:rt-multiphase]{$k$-phase}
	& $\T_1^\leq \land \bigwedge_{i > 1} (\T^\leq_i \lor \T^{>0}_{i-1})$
	& $\T_1^< \land \bigwedge_{i > 1} (\T^<_i \lor \T^{>0}_{i-1})$
	& $\bigvee_i \T^{>0}_i$ \\
\hyperref[def:rt-lex]{$k$-lexicographic}
	& $\bigwedge_{i=1}^k \big( \T^\leq_i \lor \bigvee_{j=1}^{i-1} \T^<_j \big)$
	& $\bigvee_i \T^<_i \land \bigwedge_{i=1}^{k-1}
	     \big( \T^\leq_i \lor \bigvee_{j=1}^{i-1} \T^<_j \big)$
	& $\bigwedge_i \T^{>0}_i$ \\
\hyperref[def:rt-parallel]{$k$-parallel}
	& $\bigwedge_i \T^\leq_i$
	& $\bigwedge_i \T^\leq_i \land \bigvee_i (\T^<_i \land \T^{>0}_i)$
	& $\bigvee_i \T^{>0}_i$
\end{tabular}
\end{center}
\end{definition}
\medskip

\noindent The intuition behind \autoref{def:composed-template-recipe} is that
we build composed templates recursively
using PR template recipes (C1) or $k$-piece template recipes (C2)
as the base case and
plugging them into composition rules given in (C3).

There is a composition rule
for each linear ranking template presented in \autoref{sec:ranking-templates}
but the $k$-piece ranking template and the $k$-nested ranking template.
We cannot define a $k$-piece or a $k$-nested composition rule analogously,
because not all of these ranking templates' atoms are of the form (i--iii) above:
they also have atoms
containing multiple affine-linear function symbols
($f_i(x') < f_i(x) + f_{i-1}(x)$ in \eqref{eq:rt-nested4} and
$f_j(x') < f_i(x) - \delta$ in \eqref{eq:rt-pw2}).

Given a composed template recipe $(\T^\leq, \T^<, \T^{>0})$,
we call the conjunction $\T^< \land \T^{>0}$ a \emph{composed template.}
The following theorem states that
composed templates are linear ranking templates.

\begin{theorem}
\label{thm:composed-template-recipe}
If $(\T^\leq, \T^<, \T^{>0})$ is a composed template recipe,
then the composed template $\T^< \land \T^{>0}$ is a linear ranking template.
\end{theorem}

The proof of \autoref{thm:composed-template-recipe}
is deferred to the end of this section.

\begin{example}[The $k$-Phase Composition Rule]
\label{ex:k-phase-composition-rule}
We apply the $k$-phase composition rule to $k$ PR template recipes.
Let $D := \{ \delta_i \mid 1 \leq i \leq k \}$ be parameters and
let $F = \{ f_i \mid 1 \leq i \leq k \}$
be affine-linear function symbols.
For each $i$, we have three formulas
$\T^\leq_{\text{PR},i}$, $\T^<_{\text{PR},i}$, and $\T^{>0}_{\text{PR},i}$.
Using the $k$-phase composition rule
from \autoref{def:composed-template-recipe} (C3),
we get the composed template recipe
$(\T^\leq_{k\text{-phase}}, \T^<_{k\text{-phase}}, \T^{>0}_{k\text{-phase}})$
where
\begin{align*}
        &\T^\leq_{k\text{-phase}}
\equiv f_1(x') \leq f_1(x) \;\land \bigwedge_{i=2}^k (
         f_i(x') \leq f_i(x) \;\lor\; f_{i-1}(x) > 0
       ), \\
&\begin{aligned}
        \T^<_{k\text{-phase}}
\equiv &f_1(x') < f_1(x) - \delta_1 \;\land\; \delta_1 > 0 \\
       &\land \bigwedge_{i=2}^k (
          (f_i(x') < f_i(x) - \delta_i \;\land\; \delta_i > 0)
          \;\lor\; f_{i-1}(x) > 0
        ),
\end{aligned} \\
        &\T^{>0}_{k\text{-phase}}
\equiv \bigvee_{i=1}^k f_i(x) > 0.
\end{align*}
By \autoref{thm:composed-template-recipe},
$\T^<_{k-\text{phase}} \land \T^{>0}_{k-\text{phase}}$
is a linear ranking template.
In fact, we already know this from \autoref{thm:rt-multiphase},
because the formula $\T^<_{k\text{-phase}} \land \T^{>0}_{k\text{-phase}}$
is equivalent to the $k$-phase ranking template.
\end{example}

\begin{remark}\label{rem:composing-PR-template-recipes}\hfill
\begin{enumerate}[label=(\roman*)]
\item Let $(\T^\leq, \T^<, \T^{>0})$ be the $k$-phase composition rule
	applied to $k$ PR template recipes.
	Then the composed template $\T^{>0} \land \T^<$
	is equivalent to the $k$-phase ranking template.
\item Let $(\T^\leq, \T^<, \T^{>0})$ be the $k$-lexicographic composition rule
	applied to $k$ PR template recipes.
	Then the composed template $\T^{>0} \land \T^<$ is equivalent to
	the $k$-lexicographic ranking template.
\item Let $(\T^\leq, \T^<, \T^{>0})$ be the $k$-parallel composition rule
	applied to $k$ PR template recipes.
	Then the composed template $\T^{>0} \land \T^<$ is equivalent to
	the $k$-parallel ranking template.
\end{enumerate}
\end{remark}
\begin{proof}
From \autoref{def:rt-multiphase}, \autoref{def:rt-lex},
and \autoref{def:rt-parallel}.
\end{proof}

Next, we construct a composed template
to prove termination of \autoref{ex:multilex}.

\begin{example}\label{ex:l-lex-k-phase-composition-rule}
We apply the $\ell$-lexicographic composition rule to $\ell$ copies of the
composed template recipe from \autoref{ex:k-phase-composition-rule}.
Let $D := \{ \delta_{i,j} \mid 1 \leq i \leq k, 1 \leq j \leq \ell \}$
be parameters and
let $F = \{ f_{i,j} \mid 1 \leq i \leq k, 1 \leq j \leq \ell \}$
be affine-linear function symbols.
For each $j$,
we apply the $k$-phase composition rule to $k$ PR template recipes
as in \autoref{ex:k-phase-composition-rule},
using the parameters $D_j := \{ \delta_{i,j} \mid 1 \leq i \leq k \}$ and
the affine-linear function symbols $F_j := \{ f_{i,j} \mid 1 \leq i \leq k \}$.
Let the resulting composed template recipe be denoted
$(\T^\leq_{k\text{-phase},j}, \T^<_{k\text{-phase},j}, \T^{>0}_{k\text{-phase},j})$.
Next, we apply the $\ell$-lexicographic composition rule to
the $\ell$ composed template recipes
$(\T^\leq_{k\text{-phase},j}, \T^<_{k\text{-phase},j}, \T^{>0}_{k\text{-phase},j})$
resulting in the composed template recipe
$(\T^\leq_{\text{lm}}, \T^<_{\text{lm}}, \T^{>0}_{\text{lm}})$
where
\begingroup
\allowdisplaybreaks
\begin{align*}
&\begin{aligned}
\T^\leq_{\text{lm}} \equiv
&\bigwedge_{j=1}^\ell \Bigg(
   \Big( f_{1,j}(x') \leq f_{1,j}(x)
     \land \bigwedge_{i=2}^k (f_{i,j}(x') \leq f_{i,j}(x)
       \lor f_{i-1,j}(x) > 0) \Big) \\
&\quad\lor\; \bigvee_{t=1}^{j-1}
  \Big( f_{1,t}(x') < f_{1,t}(x) - \delta_{1,t} \land \delta_{1,t} > 0 \\
&\quad\quad\quad\land \bigwedge_{i=2}^k (
       (f_{i,t}(x') < f_{i,t}(x) - \delta_{i,t} \land \delta_{i,t} > 0)
       \lor f_{i-1,t}(x) > 0) \Big)
\Bigg),
\end{aligned} \\
&\begin{aligned}
\T^<_{\text{lm}} \equiv
&\bigvee_{j=1}^\ell \Bigg(
  f_{1,j}(x') < f_{1,j}(x) - \delta_{1,j} \land \delta_{1,j} > 0 \\
&\quad\land \bigwedge_{i=2}^k (
  (f_{i,j}(x') < f_{i,j}(x) - \delta_{i,j} \land \delta_{i,j} > 0)
  \lor f_{i-1,j}(x) > 0) \Bigg) \\
&\land \bigwedge_{j=1}^{\ell-1} \Bigg(
  \Big( f_{1,j}(x') \leq f_{1,j}(x)
    \land \bigwedge_{i=2}^k (f_{i,j}(x') \leq f_{i,j}(x)
      \lor f_{i-1,j}(x) > 0) \Big) \\
&\qquad\lor\; \bigvee_{t=1}^{j-1}
  \Big( f_{1,t}(x') < f_{1,t}(x) - \delta_{1,t} \land \delta_{1,t} > 0 \\
&\qquad\qquad\land \bigwedge_{i=2}^k (
   (f_{i,t}(x') < f_{i,t}(x) - \delta_{i,t} \land \delta_{i,t} > 0)
   \lor f_{i-1,t}(x) > 0) \Big)
\Bigg),
\end{aligned} \\
&\begin{aligned}
\T^{>0}_{\text{lm}} \equiv
\bigwedge_{j=1}^\ell \bigvee_{i=1}^k f_{i,j}(x) > 0.
\end{aligned}
\end{align*}
\endgroup
By \autoref{thm:composed-template-recipe},
$\T^<_{\text{lm}} \land \T^{>0}_{\text{lm}}$
is a linear ranking template.
\end{example}

\begin{example}\label{ex:multilex2}
Using the composed template recipe from \autoref{ex:l-lex-k-phase-composition-rule},
we can find a ranking function for \autoref{ex:multilex}:
\renewcommand{\qed}{} 
\begin{align*}
f_{1,1}(q, x, y) &= 1 - x &&
f_{1,2}(q, x, y) = q \\
f_{2,1}(q, x, y) &= y &&
f_{2,2}(q, x, y) = y
\tag*{$\Diamond$} 
\end{align*}
\end{example}

\begin{proof}[Proof of \autoref{thm:composed-template-recipe}]
First, we need to check the syntactic requirements.
This was already shown for the PR ranking template
and the $k$-piece ranking template.
Any substitution is a boolean combination of parts of simpler templates,
and the syntactic requirements for linear ranking templates allow for
arbitrary boolean combinations of atoms.

To show well-foundedness, we prove the following statement by induction over
the recursive construction of the composed template recipes.
We show that
for all assignments $\nu$ to the parameters and affine-linear function symbols,
we find a function $\rho: \Sigma \to \alpha$
from the program states $\Sigma$ to some ordinal $\alpha$ such that
\begin{enumerate}[label=(\roman*)]
\item $\nu(\T^\leq)(x, x')$ implies $\rho(x') \leq \rho(x)$.
\item $\nu(\T^<)(x, x')$ and $\rho(x) > 0$ imply $\rho(x') < \rho(x)$.
\item $\nu(\T^{>0})(x, x')$ implies $\rho(x) > 0$.
\end{enumerate}
For the base case, we have a PR template recipe or
a $k$-piece template recipe,
and we get a ranking function $\rho: \Sigma \to \omega$.
Claims (ii) and (iii) follow from
\autoref{lem:ordinal-ranking} and \autoref{thm:rt-pw}.
For the PR template recipe, claim (i) holds because $f(x') \leq f(x)$ implies
$\widehat{f}(x') \leq \widehat{f}(x)$.
For the $k$-piece template recipe,
we prove this analogously to the proof of \autoref{thm:rt-pw},
using $f_i(x') \leq f_i(x)$ instead of $f_i(x') < f_i(x) - \delta$.

For the induction step, assume that claims (i--iii) hold for
the composed template recipes
$(\T^\leq_1, \T^<_1, \T^{>0}_1)$, $\ldots$, $(\T^\leq_k, \T^<_k, \T^{>0}_k)$.
Thus for every $i = 1 \ldots k$, we have a ranking function
$\rho_i: \Sigma \to \alpha_i$ with ordinal $\alpha_i$ as codomain.
We consider the three inductive cases in turn.
\begin{itemize}
\item \emph{$k$-phase}: We define the ranking function
\[
\rho(x) :=
\begin{cases}
\sum_{j=1}^{i-1} \alpha_j + \rho_i(x) & \text{if }
\rho_j(x) = 0 \text{ for all } j < i \text{ and } \rho_i(x) > 0, \\
0 & \text{otherwise.}
\end{cases}
\]
Let $(x, x') \in \nu(\T^\leq)$,
and let $i$ be the current phase, i.e.,
$\rho_i(x) > 0$ and $\rho_j(x) = 0$ for all $j < i$.
For all $j < i$, we have
$(x, x') \in \nu(\T^\leq_j)$ and $(x, x') \notin \nu(\T^{>0}_j)$ by inductive hypothesis,
and hence $\rho_j(x') \leq \rho_j(x) = 0$.
Therefore we obtain $(x, x') \in \nu(\T^\leq_i)$ and hence $\rho_i(x') \leq \rho_i(x)$
(note the subscript $i$ instead of $j$),
which implies $\rho(x') \leq \rho(x)$ in case $\rho_i(x') > 0$.
Otherwise we have a phase transition and thus
$\rho(x') = 0 < \rho(x)$ or $i < k$.
In the latter case, we know $\rho_{i+1}(x') < \alpha_{i+1}$,
and hence
$
  \rho(x)
> \sum_{j=1}^{i-1} \alpha_j
> \sum_{j=1}^{i-2} \alpha_j + \rho_{i-2}(x')
= \rho(x')
$.

For $(x, x') \in \nu(\T^<)$, we analogously get
$\rho_i(x') < \rho_i(x)$ and hence $\rho(x') < \rho(x)$.
For $(x, x') \in \nu(\T^{>0})$, we have that
$\rho_i(x) > 0$ for some $i$ and hence $\rho(x) > 0$ by the induction hypothesis.

\item \emph{$k$-lexicographic}: We define the ranking function
\[
\rho(x) := \sum_{i=1}^k \rho_i(x) \prod_{j=i+1}^k \alpha_j.
\]
Let $(x, x') \in \nu(\T^\leq)$.
If $(x, x') \in \nu(\T^\leq_i)$ for all $i$,
then we get $\rho_i(x') \leq \rho_i(x)$ by the induction hypothesis, and hence
$\rho(x') \leq \rho(x)$.
Otherwise there is an $n \leq k$ such that $(x, x') \in \nu(\T^<_n)$ and
$(x, x') \in \nu(\T^\leq_j)$ for all $j < n$.
By the induction hypothesis, we obtain $\rho_n(x') < \rho_n(x)$ and
$\rho_j(x') \leq \rho_j(x)$ for all $j < n$.
Since $\rho_i(x') < \alpha_i$, we have
$\prod_{j=n+1}^k \alpha_j > \sum_{i=n+1}^k \rho_i(x') \prod_{j=i+1}^k \alpha_j$
and thus
\begingroup
\allowdisplaybreaks
\begin{align*}
\rho(x)
&=    \sum_{i=1}^k \rho_i(x) \prod_{j=i+1}^k \alpha_j \\
&\geq \left( \sum_{i=1}^{n-1} \rho_i(x) \prod_{j=i+1}^k \alpha_j \right)
      + \left( \rho_n(x) \prod_{j=n+1}^k \alpha_j \right) \\
&=    \left( \sum_{i=1}^{n-1} \rho_i(x) \prod_{j=i+1}^k \alpha_j \right)
      + \left( (\rho_n(x) - 1) \prod_{j=n+1}^k \alpha_j \right)
      + \prod_{j=n+1}^k \alpha_j \\
&>    \left( \sum_{i=1}^{n-1} \rho_i(x) \prod_{j=i+1}^k \alpha_j \right)
      + \left( (\rho_n(x) - 1) \prod_{j=n+1}^k \alpha_j \right)
      + \left( \sum_{i=n+1}^k \rho_i(x') \prod_{j=i+1}^k \alpha_j \right) \\
&\geq \left( \sum_{i=1}^{n-1} \rho_i(x') \prod_{j=i+1}^k \alpha_j \right)
      + \left( \rho_n(x') \prod_{j=n+1}^k \alpha_j \right)
      + \left( \sum_{i=n+1}^k \rho_i(x') \prod_{j=i+1}^k \alpha_j \right) \\
&=    \rho(x').
\end{align*}
\endgroup
For $(x, x') \in \nu(\T^<)$, we proceed analogously except that the case
$(x, x') \in \nu(\T^\leq_i)$ for all $i$ cannot occur.
For $(x, x') \in \nu(\T^{>0})$,
we get $\rho_i(x) > 0$ by the induction hypothesis,
thus $\rho(x) > 0$.

\item \emph{$k$-parallel}: We define the ranking function
\[
\rho(x) := \sum_{i=1}^k \rho_i(x).
\]
For $(x, x') \in \nu(\T^\leq)$,
we have that $\rho_i(x') \leq \rho_i(x)$ for all $i$ by the induction hypothesis,
and hence $\rho(x') \leq \rho(x)$.
For $(x, x') \in \nu(\T^<)$, we again have
by the induction hypothesis that
 $\rho_i(x') \leq \rho_i(x)$ for all $i$,
and that there is an $i$ such that $\rho_i(x') < \rho_i(x)$.
Therefore we obtain $\rho(x') < \rho(x)$.
For $(x, x') \in \nu(\T^{>0})$, we have that
there is an $i$ such that $\rho_i(x) > 0$ by the induction hypothesis,
therefore $\rho(x) > 0$.
\end{itemize}
This completes the induction.
To finish the proof, we note that according to claim (ii) and (iii),
$\rho$ is a ranking function for $\nu(\T^< \land \T^{>0})$,
and by \autoref{lem:rf} this implies that
$\nu(\T^< \land \T^{>0})$ is well-founded.
\end{proof}

Although the procedure introduced in this section allows for
an infinite number of different ranking templates,
it is not exhaustive.
We expect that there are many more types of ranking functions
that can be formalized using linear ranking templates,
and possibly also composed with other templates.


\section{Synthesizing Ranking Functions}
\label{sec:synthesizing-ranking-functions}

Following related approaches~\cite{ADFG10,BMS05linrank,BMS05polyrank,CSS03,HHLP13,PR04,Rybalchenko10,SSM04},
we transform the $\exists\forall$-con\-straint \eqref{eq:rt} into an $\exists$-constraint.
This transformation makes the constraint more easily solvable
not only because we remove universal quantification,
but also because it reduces the number of nonlinear operations in the constraint.
Every application of an affine-linear function symbol $f$
corresponds to a nonlinear term $\tr{s_f} x + t_f$ where
$s_f$ is a vector of real-valued parameters and
$t_f$ is a real-valued parameter.
For this step, we need the following theorem.

\subsection{Motzkin's Transposition Theorem}
\label{ssec:motzkin}

Intuitively, Motzkin's transposition theorem states that
a given system of linear inequalities has no solution if and only if
a contradiction can be derived
via a positive linear combination of the inequalities.

\begin{theorem}[{Motzkin's Transposition Theorem~\cite[Cor.\ 7.1k]{Schrijver99}}]
\label{thm:Motzkin}
For $A \in \mathbb{K}^{m \times n}$, $C \in \mathbb{K}^{\ell \times n}$,
$b \in \mathbb{K}^m$, and $d \in \mathbb{K}^\ell$,
the formulas \eqref{eq:motzkin1} and \eqref{eq:motzkin2} are equivalent.
\begin{align}
&\hspace{14.7mm} \forall x \in \mathbb{K}^n.\;
  \neg (Ax \leq b \;\land\; Cx < d)
\tag{M1}\label{eq:motzkin1} \\[1.5mm]
&\begin{aligned}
\exists \lambda \in \mathbb{K}^m \;\exists \mu \in \mathbb{K}^\ell.\;
  &\lambda \geq 0 \;\land\; \mu \geq 0 \\
  \land\; &\tr{\lambda} A + \tr{\mu} C = 0
  \;\land\;  \tr{\lambda} b + \tr{\mu} d \leq 0 \\
  \land\; &( \tr{\lambda} b < 0 \;\lor\; \mu \neq 0 )
\end{aligned}\tag{M2}\label{eq:motzkin2}
\end{align}
\end{theorem}

\noindent If $\ell$ is set to $1$ in \autoref{thm:Motzkin}, we obtain
the affine version of Farkas' lemma~\cite[Cor.\ 7.1h]{Schrijver99}.
Therefore \hyperref[thm:Motzkin]{Motzkin's theorem} is strictly
superior to Farkas' lemma, as it allows for a combination of both
strict and non-strict inequalities.  Moreover, it is logically optimal
in the sense that it enables the transformation of \emph{any} purely
universally quantified ($\Pi_1^0$) formula from the theory of linear
arithmetic.


\subsection{Constraint Transformation}
\label{ssec:constraint-transformation}

We fix a linear loop program $\loopt$ and a linear ranking template $\T$
with parameters $D$ and affine-linear function symbols $F$.
For simplicity of presentation,
we assume the loop program $\loopt$ does not contain any strict inequalities,
and the ranking template $\T$ does not contain any non-strict inequalities;
however, recall that we are using \hyperref[thm:Motzkin]{Motzkin's theorem}
instead of Farkas' lemma precisely to lift this restriction.
For the fully general constraints, see \cite[Ch.\ 5]{Leike13}.
We write $\loopt$ in disjunctive normal form and $\T$ in conjunctive normal form:
\begin{align*}
        \loopt(x, x')
&\equiv \bigvee_{i \in I} A_i \abovebelow{x}{x'} \leq b_i \\
        \T(x, x')
&\equiv \bigwedge_{j \in J} \bigvee_{\ell \in L_j} \T_{j,\ell}(x, x')
 \equiv \bigwedge_{j \in J} \bigvee_{\ell \in L_j}
          \tr{t_{j,\ell}} \abovebelow{x}{x'} > e_{j,\ell}
\end{align*}
We prove the termination of $\loopt$ by solving the constraint \eqref{eq:rt}.
This constraint is implicitly existentially quantified over the parameters $D$
and the parameters corresponding to the affine-linear function symbols $F$.
\begin{equation}\label{eq:constraint}
\forall x, x'.\;
\left(
\Big( \bigvee_{i \in I} A_i \abovebelow{x}{x'} \leq b_i \Big)
\rightarrow
\Big( \bigwedge_{j \in J} \bigvee_{\ell \in L_j} \T_{j,\ell}(x, x') \Big)
\right)
\end{equation}
First, we transform the constraint \eqref{eq:constraint} into an equivalent constraint of the form required by \hyperref[thm:Motzkin]{Motzkin's theorem}.
\begin{equation}\label{eq:constraint2}
\bigwedge_{i \in I} \bigwedge_{j \in J}
\forall x, x'.\;
\neg \left(
A_i \abovebelow{x}{x'} \leq b_i
\;\land\;
\Big( \bigwedge_{\ell \in L_j} \neg\T_{j,\ell}(x, x') \Big)
\right)
\end{equation}
Now, \hyperref[thm:Motzkin]{Motzkin's Transposition theorem} transforms
the constraint \eqref{eq:constraint2} into
an equivalent existentially quantified constraint:
\begin{equation}\label{eq:transformed-constraint}
\bigwedge_{i \in I} \bigwedge_{j \in J}
\exists \lambda \geq 0\; \exists \zeta \geq 0.\;
\tr{\lambda} A_i + \sum_{\ell \in L_j} \zeta_\ell \tr{t_{j,\ell}} = 0
\;\land\;
\tr{\lambda} b_i + \sum_{\ell \in L_j} \zeta_\ell e_{j,\ell} < 0
\end{equation}
For every inequality in \eqref{eq:motzkin1},
a new existentially quantified variable
is added in \eqref{eq:motzkin2}.
These new existentially quantified variables are called
\emph{Motzkin coefficients}.

The $\exists$-constraint \eqref{eq:transformed-constraint}
is then checked for satisfiability.
If an assignment is found, it gives rise to a ranking function.
Conversely, if no assignment exists,
then there cannot be an instantiation of the linear ranking template and
thus no ranking function of the kind formalized by the linear ranking template exists.
In this sense our method is sound and complete.

\begin{theorem}[Soundness]\label{thm:soundness}
If the transformed $\exists$-constraint \eqref{eq:transformed-constraint}
is satisfiable, then the linear loop program terminates.
\qed
\end{theorem}

\begin{theorem}[Completeness]\label{thm:completeness}
If the $\exists\forall$-constraint \eqref{eq:rt} is satisfiable, then so is
the transformed $\exists$-constraint \eqref{eq:transformed-constraint}.
\qed
\end{theorem}


\subsection{Ranking Template Pools}
\label{ssec:rt-pools}

Our method for ranking function synthesis can be applied as follows.
We fix a finite pool of linear ranking templates $\mathcal{T}$,
consisting of
multiphase, nested, piecewise, lexicographic, and parallel ranking templates
as well as composed templates in various sizes.
The input is a linear loop program $\loopt$
that we want to check for termination.
We start by picking a linear ranking template $\T$ from the pool $\mathcal{T}$.
From the ranking template $\T$ we build the constraint \eqref{eq:rt}
to the parameters and affine-linear function symbols of $\T$.
This constraint is transformed using \hyperref[thm:Motzkin]{Motzkin's theorem}
to an $\exists$-constraint \eqref{eq:transformed-constraint}.
If this constraint is satisfiable,
this gives rise to a ranking function according to \autoref{lem:rf},
and thus we proved that the loop program $\loopt$ terminates.
Otherwise, we try again
using the next linear ranking template from the pool $\mathcal{T}$
until the pool has been exhausted.
If the pool has been exhausted,
the proof of the loop program $\loopt$'s termination failed.
However, due to the completeness of our method,
we know that the loop program $\loopt$
does not have a ranking function of the form
specified by any of the linear ranking templates in the pool.
\autoref{fig:synthesis-method} is a description of our method in pseudocode.

\begin{figure}[t]
\begin{description}
\item[Input] linear loop program $\loopt$ and a list of linear ranking templates $\mathcal{T}$
\item[Output] a ranking function for $\loopt$ or \texttt{null} if none is found
\end{description}
\begin{center}
\begin{minipage}{98mm}
\begin{lstlisting}
foreach $\T \in \mathcal{T}$ do:
    let $\varphi$ = $\forall x, x'.\; \big( \loopt(x, x') \rightarrow \T(x, x') \big)$
    let $\psi$ = transformWithMotzkin($\varphi$)
    if SMTsolver.checkSAT($\psi$):
        let ($D$, $F$) = $\T$.getParameters()
        let $\nu$ = getAssignment($\psi$, $D$, $F$)
        return $\T$.extractRankingFunction($\nu$)
return $\texttt{null}$
\end{lstlisting}
\end{minipage}
\end{center}
\caption{
Our ranking function synthesis algorithm described in pseudocode.
The function \texttt{transformWithMotzkin} transforms
the $\exists\forall$-constraint $\varphi$ into an $\exists$-constraint $\psi$
as described in \autoref{ssec:constraint-transformation}.
The ranking function is extracted from an assignment of the template
with the function \texttt{extractRankingFunction}.
This function returns a description of the ranking function depending on the template;
for example the ordinal-based representation from the proofs.
}\label{fig:synthesis-method}
\end{figure}


\section{Linear Lasso Programs}
\label{sec:linear-lasso-programs}

\begin{wrapfigure}{r}{0.4\textwidth}
\vspace{-7mm}
\begin{center}
\begin{tikzpicture}
\node (1) at (0,0) {};
\node[state] (2) at (2.5,0) {};
\draw[arrows={-triangle 45}] (1) to node[above] {$\stemt$} (2);
\draw[arrows={-triangle 45}, loop right,in=30,out=-30,looseness=10]
  (2) to node {$\loopt$} (2);
\end{tikzpicture}
\end{center}
\caption{A lasso program.
}
\label{fig:lasso}
\end{wrapfigure}

Our method extends to
the more general setting of \emph{linear lasso programs}.
These are linear loop programs that have a program stem in addition to the loop
(see \autoref{fig:lasso}).
We use affine-linear inductive invariants
to extract the information that is crucial for the termination proof
from the stem.
This is in line with related approaches~\cite{CSS03,SSM04,BMS05linrank,HHLP13}.

\begin{definition}[Linear Lasso Program]
\label{def:linear-lasso-program}
A \emph{linear lasso program} $\prog = (\stemt, \loopt)$ consists of
\begin{itemize}
\item a linear loop program $\loopt$, and
\item a predicate $\stemt$, defined by
  a formula with the free variables $x$ of the form
\[
\bigvee_{i \in I} \big( A_i x \leq b_i \;\land\; C_i x < d_i \big)
\]
for some finite index set $I$, some matrices $A_i \in \mathbb{K}^{n \times m_i}$,
$C_i \in \mathbb{K}^{n \times k_i}$,
and some vectors $b_i \in \mathbb{K}^{m_i}$ and $d_i \in \mathbb{K}^{k_i}$.
\end{itemize}
The linear lasso program $\prog$ is called \emph{conjunctive}
iff there is only one disjunct in both transitions $\stemt$ and $\loopt$.
\end{definition}

\begin{definition}[Affine-Linear Supporting Invariant]
\label{def:supporting-invariant}
A formula $\psi$ is an \emph{affine-linear supporting invariant}
for the linear lasso program $\prog$ iff
there is an affine-linear function $f$ such that
\[
\psi(x) \equiv f(x) \rhd 0
\]
with $\rhd \in \{ \geq, > \}$,
and the following two formulas 
hold.
\begin{align}
\forall x.&\; \stemt(x) \rightarrow \psi(x)
\label{eq:ii}\tag{II} \\
\forall x, x'.&\; \psi(x) \land \loopt(x, x') \rightarrow \psi(x')
\label{eq:ic}\tag{IC}
\end{align}
The affine-linear supporting invariant $\psi$ is \emph{strict}
iff $\rhd$ is $>$ and \emph{non-strict} otherwise.
\end{definition}

Given a linear lasso program,
we do the same transformation steps as in \autoref{ssec:constraint-transformation},
adding a finite number of supporting invariants:
\[
\forall x, x'.\; \loopt(x, x') \;\land\; \bigwedge_\ell \psi_\ell(x)
  \rightarrow \T(x, x')
\]
In fact, every conjunct in \eqref{eq:constraint2} gets $m$ supporting invariants:
\begin{equation}\label{eq:constraint2-lasso}
\bigwedge_{i \in I} \bigwedge_{j \in J}
\forall x, x'.\;
\neg \left(
A_i \abovebelow{x}{x'} \leq b_i
\;\land\;
\Big( \bigwedge_{\ell=1}^m \psi_{i,j,\ell}(x) \Big)
\;\land\;
\Big( \bigwedge_{\ell \in L_j} \neg\T_{j,\ell}(x, x') \Big)
\right)
\end{equation}
To insure that the $\psi_{i,j,\ell}(x)$ are indeed supporting invariants,
we add the constraints \eqref{eq:ii} and \eqref{eq:ic}
for each $(i, j, \ell) \in I \times J \times \{ 1, \ldots, m \}$.
Each of these constraints is then transformed using \hyperref[thm:Motzkin]{Motzkin's theorem}.
analogously to \autoref{ssec:constraint-transformation}.
Not all invariants are inductive and
we only consider invariants that are affine-linear inequalities.
We do not retain completeness of
our method in the sense of \autoref{thm:completeness}.

The invariant initiating \eqref{eq:ii} is a linear constraint,
but the invariant consecution \eqref{eq:ic} is nonlinear.
We could make \eqref{eq:ic} linear
by restricting ourselves to non-decreasing invariants~\cite{HHLP13}.
However, the overall constraints are generally still nonlinear
because the constraints that come from the linear ranking template
are generally nonlinear.


\section{Related Work}
\label{sec:related-work}

Synthesis of linear ranking functions for linear loop programs
was first discussed by Colón and Sipma~\cite{CS01}.
This was extended to a complete template-based method
by Podelski and Ry\-bal\-chen\-ko~\cite{PR04, Rybalchenko10},
using the PR ranking template as discussed in \autoref{ex:rt-pr}.
Their method is not complete over the integers.
Cook et al.~\cite{CKRW13} compute the integral hull
of transition relations
in order obtain the same completeness
for integers and bitvectors.
Bagnara and Mesnard generalize the PR ranking template to
the 2-phase ranking template,
relying on nonlinear constraint solving~\cite{BM13}.

Bradley, Manna, and Sipma propose a constraint-based approach for linear lasso
programs~\cite{BMS05linrank}.
Their termination argument is a lexicographic ranking function
with each lexicographic component corresponding to one loop disjunct.
This requires nonlinear constraint solving and
an ordering on the loop disjuncts.
The authors extend this approach in \cite{BMS05polyrank}
by the use of \emph{template trees}.
These trees allow each lexicographic component to have a ranking function that
decreases not necessarily in every step, but \emph{eventually}.



Ben-Amram and Genaim discuss
the synthesis of affine-linear and lexicographic ranking functions
for linear loop programs over the integers~\cite{BAG13}.
They prove that this problem is generally co-NP-complete and show that
several special cases admit a polynomial time complexity.

In \cite{ChenFM12} the authors also address the problem of
finding termination arguments for
(not necessarily conjunctive) linear loop programs. 
In contrast to our work,
the authors do not synthesize the termination argument directly.
Instead, they iteratively synthesize linear ranking functions and
obtain a disjunctively well-founded relation~\cite{PR04TI}
as a termination argument.

Approaches for computing lexicographic linear ranking functions
for a more general class of programs,
namely programs that can consist of several (potentially nested) loops
are presented in \cite{ADFG10} and \cite{CookSZ13}.
On linear loop programs,
both algorithms involve choosing an ordering on the loop disjuncts.
Hence, both approaches are either incomplete or
have to use backtracking
to iteratively consider all possible orderings of loop disjuncts.

Our method is not able to prove termination
for all terminating linear loop programs.
Termination is decidable for the subclass of
deterministic conjunctive linear loop programs of the form
\begin{center}
\texttt{while($B_s x > b_s \land B_w x \geq b_w$) $x$ := $Ax + c$;}
\end{center}
where the matrices $B_s$, $B_w$, $A$ and
vectors $b_s$, $b_w$, $c$ are rational, and
variables can take on rational or real values~\cite{Tiwari04}.
This class also admits decidable termination analysis over the integers
for the homogeneous case where $b_s, b_w, c = 0$~\cite{Braverman06}.
However, their method is not targeted at the synthesis of ranking functions.

Ranking functions can also be computed
via abstract interpretation~\cite{CousotC12}.
Urban and Miné~\cite{Urban13,UM14ESOP,UM14SAS} introduced
the domain of piecewise defined ordinal-valued functions for this approach.
In contrast to our work,
their approach is applicable to programs with arbitrary structure and
not restricted to linear lasso programs.
However, the authors do not provide completeness results
that state that a ranking function of a certain form can always be found.


\section{Conclusion}
\label{sec:conclusion}

\begin{table}[t]
\begin{center}
\renewcommand{\arraystretch}{1.2} 
\begin{tabular}{lcccccc}
	& PR & $k$-phase & $k$-nested & $k$-piece & $k$-lexicographic & $k$-parallel \\
\hline
Parameters
	& $1$ & $k$      & $1$      & $1$             & $k$            & $k$ \\
Function symbols
	& $1$ & $k$      & $k$      & $2k$            & $k$            & $k$ \\
Conjuncts
	& $3$ & $2k + 1$ & $k + 2$  & $k^2 + k + 2$   & $3k$           & $2^k + 2k$ \\
Atoms
	& $3$ & $4k - 1$ & $k + 2$  & $3k^2 + 2k + 1$ & $(5k^2 + k)/2$ & $k2^k + 2k$
\end{tabular}
\end{center}
\caption{
Statistics of our linear ranking templates in CNF;
the integer $k$ specifies their size.
Every affine-linear function symbol
constributes $n + 1$ parameters to the template,
where $n$ is the number of program variables.
}
\label{tab:rt}
\end{table}

We presented a sound and complete method for constraint-based synthesis of ranking functions for linear loop programs.
For this method, we introduced the notion of \emph{linear ranking templates},
which are parameterized formulas for well-founded relations.
In \autoref{sec:templates} we established how they can be applied to prove termination (\autoref{lem:termination}) and
that an instantiation of a linear ranking template gives rise to a ranking function (\autoref{lem:rf}).
Our method can be applied to different kinds of ranking functions that previously have been considered independently (affine-linear and lexicographic ranking functions),
in addition to enabling new kinds
(multiphase, piecewise, and parallel ranking functions).
The ranking templates can also be composed into more powerful templates,
allowing for more general ranking functions.
See \autoref{tab:rt} for statistics on the size of our ranking templates.

Our method can be applied to linear loop programs and linear lasso programs
with variables that are rational numbers, real numbers, or integers.
In general, it requires solving nonlinear algebraic constraints,
but some linear ranking templates such as
the PR ranking template or the nested ranking template
only require linear constraint solving.

\subsection*{Acknowledgements}
We wish to thank 
Samir Genaim for pointing our an error
in the conference version of \autoref{thm:conjunctive-loop-no-multiphase},
and Amir M.\ Ben-Amram for detailed comments on the Master's thesis~\cite{Leike13}
from which this paper was derived.
Moreover, we thank Andreas Podelski
for his detailed feedback and helpful suggestions.


\bibliographystyle{alpha}
\bibliography{references}

\end{document}